\documentclass[sn-mathphys]{sn-jnl}

\jyear{2021}%

\theoremstyle{thmstyleone}%
\newtheorem{theorem}{Theorem}
\newtheorem{proposition}[theorem]{Proposition}
\newtheorem{lemma}{Lemma}%
\newtheorem{corollary}{Corollary}
\theoremstyle{thmstyletwo}%
\newtheorem{example}{Example}%
\newtheorem{remark}{Remark}%
\theoremstyle{thmstylethree}%
\newtheorem{definition}{Definition}%

\begin{document}

\title[Bounds on Binary Niederreiter-Rosenbloom-Tsfasman LCD codes]{Bounds on Binary Niederreiter-Rosenbloom-Tsfasman LCD codes}


\author{\fnm{Welington} \sur{Santos}}\email{welington@vt.edu}

\affil{\orgdiv{Department of Mathematics} \orgname{Virginia Tech},
\orgaddress{\street{225 Stanger Street},
\city{Blacksburg}, \postcode{204061}, \state{VA}, \country{USA}}}


\abstract{Linear complementary dual codes (LCD codes) are codes whose intersections with their dual codes are trivial. These codes were introduced by Massey in 1992. LCD codes have wide applications in data storage, communication systems and cryptography. Niederreiter-Rosenbloom-Tsfasman LCD codes (NRT-LCD codes) were introduced by Heqian, Guangku and Wei as a generalization of LCD codes for the NRT metric space $M_{n,s}(\mathbb{F}_{q})$. In this paper, we study LCD$[n\times s,k]$, the maximum minimum NRT distance among all binary $[n\times s,k]$ NRT-LCD codes. We prove the existence (non-existence) of binary maximum distance separable NRT-LCD codes in $M_{1,s}(\mathbb{F}_{2})$. We present a linear programming bound for binary NRT-LCD codes in $M_{n,2}(\mathbb{F}_{2})$. We also give two methods to construct binary NRT-LCD codes.}




\maketitle

\section{Introduction}\label{sec1}

An $[s,k]$ linear code $\mathcal{C}$ is called a Euclidean linear complementary dual (LCD for short) code if $\mathcal{C}\cap\mathcal{C}^{\perp_{E}}=\lbrace 0_{s}\rbrace$, where $\mathcal{C}^{\perp_{E}}$ is the Euclidean dual code of $\mathcal{C}$ and $0_{s}$ is the zero vector of length $s$. LCD codes were introduced by Massey \cite{Massey1992}, where LCD codes were shown to provide an optimum linear coding solution for the two-user binary adder channel and applied to the nearest-neighbor decoding problem. Moreover, in \cite{Massey1992} it was proved that asymptotically good LCD codes exist. In recent years, there has been an increasing interest in LCD codes for both theoretical and practical reasons (see \cite{Carlet2016,Carlet2019,Araya2019,Araya2020,Araya2020-2}). In particular, Carlet and Guilley showed in \cite{Carlet2016} that LCD codes are important in information protection and armoring implementations against side-channel attacks and fault non-invasive attacks. An important problem is to determine the largest minimum weight among all $[s,k]$ codes in a certain class of codes from a given pair $(s,k)$.

Coding theory has also been developed with respect to alternative metrics. One of those metrics is the Niederreiter Rosenbloom-Tsfasman (NRT for short) metric which was introduced in \cite{NRT} to model transmission over a set of parallel channels subject to fading. Independently, Niederreiter \cite{H.Niederreiter} pursued a maximization problem in finite vector spaces which turned out to be equivalent to coding theory problems in Niederreiter Rosenbloom-Tsfasman metric spaces, as was shown by Brualdi, Graves, and Lawrence \cite{brualdi}.

The NRT metric is more appropriate for the correction of random errors that occur with high probability in the first coordinates. Since its definition several coding-theoretic questions with respect to the Niederreiter Rosenbloom-Tsfasman metric have been investigated, such as bounds \cite{BargAndPurka,Quistorff}, weight distribution and MacWilliams Identities \cite{DS,Sharma,Siap}, self-dual codes \cite{Marka,InvariantSantos}, MDS codes \cite{4,MDS_Dougherty,Skriganov2001uniform-distributions}, burst error enumeration \cite{Jain,I.Siap}, BCH codes \cite{Zhou}, and decoding \cite{Ozen, Nielsen,FractionalNRTSantos}. 

Recently, Heqian, Guangkui, and Wei defined NRT-LCD codes in \cite{H.Xu}. In the same paper, the basic conditions for an NRT metric code be NRT-LCD are studied. The aim of this paper is to study $LCD[n\times s,k]$ which is the maximum of the possible minimum NRT distance among all binary $[n\times s,k]$ NRT-LCD codes $\mathcal{C}\subseteq M_{n,s}(\mathbb{F}_{2})$. 

This paper is organized as follows. Section \ref{Preliminaries} contains some basic facts and definitions about codes in the Niederreiter Rosenbloom-Tsfasman metric space and NRT-LCD codes. In Section \ref{1xs}, we study $LCD[1\times s,k]$. In Section \ref{nx2}, we present a linear programming bound for NRT-LCD codes in $M_{n,2}(\mathbb{F}_{2})$ and study $LCD[n\times 2,2]$. Finally, in Section \ref{constructions}, we present a new construction of binary NRT-LCD codes in $M_{1,s}(\mathbb{F}_{2})$ and a new construction of binary NRT-LCD codes in $M_{n,s}(\mathbb{F}_{2})$.

\section{Preliminaries}\label{Preliminaries}
\subsection{Niederreiter Rosenbloom-Tsfasman metric codes}

Let $M_{n,s}(\mathbb{F}_{q})$ denote the set of all $n\times s$ matrices over the finite field $\mathbb{F}_q$. 

Given a matrix $v=[v_{i,j}]_{n\times s}\in M_{n,s}(\mathbb{F}_{q})$, we write $v=[v_1;v_2;\ldots;v_n]$ or $v=[w_{1}\mid w_{2}\mid\cdots\mid w_{s}]$, where $v_{i}=[v_{i,1},\ldots,v_{i,s}]$ is the $i$-th row of $v$, and $w_{j}=[w_{1,j},w_{2,j},\ldots;w_{n,j}]$ is the $j$-th column of $v$.

The Niederreiter-Rosenbloom-Tsfasman (for short, NRT) weight of $v$ is given by

\begin{equation}
\rho(v)=\sum_{i=1}^{n}\rho(v_{i})
\end{equation} 
where,
\begin{eqnarray}
\rho(v_{i})&=\left\lbrace\begin{array}{cl}
\max\left\lbrace j:v_{i,j}\neq 0\right\rbrace&\text{if}\ v_{i}\neq 0\\
0&\text{if}\ v_{i}=0.
\end{array}\right.
\end{eqnarray} 
\noindent for $i=1,2,\ldots,n$.

The NRT metric is defined by $d_{N}(u,v)=\rho(u-v)$, where $u,v\in M_{n,s}(\mathbb{F}_{q})$. We note that for $s=1$, that is, for $v\in M_{n,1}\cong\mathbb{F}^{n}_{q}$ the NRT metric becomes the usual Hamming metric.

Given a subset $\mathcal{C}\subseteq M_{n,s}(\mathbb{F}_{q})$ the minimum NRT distance of $\mathcal{C}$ is defined by:
\begin{equation}
d_{N}(\mathcal{C})=\min\left\lbrace d_{N}(v,u):d_{N}(v,u)\neq 0\right\rbrace.
\end{equation}
We sometimes write $d_{N}$ rather than $d_{N}(\mathcal{C})$ if the code is clear from the context. 

Given $\mathcal{C}\subseteq M_{n,s}(\mathbb{F}_{q})$ it is easy to check that $d_{H}(\mathcal{C})\leq d_{N}(\mathcal{C})$, where $d_{H}$ is the minimum Hamming distance of $\mathcal{C}$.

\begin{definition}
An $[n\times s,k,d_{N}]$ NRT metric code is a linear subspace $\mathcal{C}\subseteq M_{n,s}(\mathbb{F}_{q})$ of dimension $k$ and minimum NRT distance $d_{N}$. 
\end{definition}
Note that the minimum distance of NRT metric code $\mathcal{C}\subseteq M_{n,s}(\mathbb{F}_{q})$ is given by $d_{N}(\mathcal{C})=\min\lbrace\rho(v):v\in\mathcal{C},v\neq 0\rbrace$. Moreover, we have the following singleton bound for $[n\times s,k,d_{N}]$ NRT metric codes 
\begin{equation}\label{Singletonbound}
 d_{N}(\mathcal{C})\leq ns-k+1.   
\end{equation}
An NRT metric code with $d_{N}(\mathcal{C})=ns-k+1$ is called maximum distance separable code.
\subsection{Niederreiter-Rosenbloom-Tsfasman LCD codes}
Niederreiter-Rosenbloom-Tsfasman linear complementary dual (NRT-LCD) codes were defined in \cite{H.Xu}. In order to study NRT-LCD codes, we introduce the following definitions.
\begin{definition}
The dual code $\mathcal{C}^{\perp_{N}}$ of an $[n\times s,k,d_{N}]$ NRT metric code $\mathcal{C}$ is defined as $$\mathcal{C}^{\perp_{N}}=\lbrace u\in M_{n,s}(\mathbb{F}_{q})\mid\langle u,v\rangle_{N}=0\text{ for all } v\in\mathcal{C}\rbrace,$$ 
where for $u=[u_1;\ldots;u_{n}],v=[v_{1};\ldots;v_{n}]\in M_{n,s}(\mathbb{F}_{q})$
\begin{equation}
\langle u,v\rangle_{N}=\sum_{i=1}^{n}\langle u_{i},v_{i}\rangle_{N}\text{, }\langle u_{i},v_{i}\rangle_{N}=\sum_{j=1}^{s}u_{i,j}u_{i,s-j+1}.
\end{equation}
\end{definition}

\begin{definition}(\cite{H.Xu})
An $[n\times s,k,d_{N}]$ NRT metric code $\mathcal{C}$ is called \textit{NRT linear complementary dual code (NRT-LCD code for short)} if $\mathcal{C}\cap\mathcal{C}^{\perp_{N}}=\lbrace 0\rbrace$.  
\end{definition}

\begin{definition}(\cite{Marka})
Let $A=\left[a_{ij}\right]\in M_{k,s}(\mathbb{F}_{q})$. The flip of $A$ is defined by $flip(A)=[a_{il}]$, where $l=s-j+1$ for $1\leq i\leq k$ and $1\leq j\leq s$. The transpose of  $flip(A)$ is denoted by $A^{\circ}$.
\end{definition}

\begin{definition}(\cite{InvariantSantos})
Let $A=[A_{1}\mid A_{2}\mid\cdots\mid A_{n}]\in M_{k,ns}(\mathbb{F}_{q})$, where $A_{i}$ is a $k\times s$ matrix for $i=1,2,\ldots,n$. The ordered flip of $A$ is defined by
\begin{equation*}
Oflip(A)=[flip(A_{1})\mid flip(A_{2})\mid\cdots\mid flip(A_{n})].
\end{equation*}
The transpose of $Oflip(A)$ is denoted by
\begin{equation}
A^{\dag}=[Oflip(A)]^{T}=\left[\begin{array}{c}
 A^{\circ}_{1}\\
 A^{\circ}_{2}\\
 \vdots\\
 A^{\circ}_{n}
\end{array}\right].
\end{equation}
Note that, if $n=1$ then, $A^{\dag}=A^{\circ}$.
\end{definition}

\begin{example}
Let $A\in M_{3,4}(\mathbb{F}_{q})$ be any matrix given by
\begin{equation*}
A=\left[\begin{array}{cccc}
a_{1,1}&a_{1,2}&a_{1,3}&a_{1,4}\\
a_{2,1}&a_{2,2}&a_{2,3}&a_{2,4}\\
a_{3,1}&a_{3,2}&a_{3,3}&a_{3,4}
\end{array}\right].
\end{equation*}
Then, $flip(A)$ and $A^{\circ}$ are given by
\begin{equation*}
flip(A)=\left[\begin{array}{cccc}
a_{1,4}&a_{1,3}&a_{1,2}&a_{1,1}\\
a_{2,4}&a_{2,3}&a_{2,2}&a_{2,1}\\
a_{3,4}&a_{3,3}&a_{3,2}&a_{3,1}
\end{array}\right]\text{, }A^{\circ}=\left[\begin{array}{ccc}
a_{1,4}&a_{2,4}&a_{3,4}\\
a_{1,3}&a_{2,3}&a_{3,3}\\
a_{1,2}&a_{2,2}&a_{3,2}\\
a_{1,1}&a_{2,1}&a_{3,1}
\end{array}\right].
\end{equation*}
Let $A\in M_{3,4}(\mathbb{F}_{q})$ be any matrix write as
\begin{equation*}
A=\left[\begin{array}{cc|cc}
a_{1,1}&a_{1,2}&a_{1,3}&a_{1,4}\\
a_{2,1}&a_{2,2}&a_{2,3}&a_{2,4}\\
a_{3,1}&a_{3,2}&a_{3,3}&a_{3,4}
\end{array}\right].
\end{equation*}
Then, $Oflip(A)$ and $A^{\dag}$ are given by
\begin{equation*}
Oflip(A)=\left[\begin{array}{cc|cc}
a_{1,2}&a_{1,1}&a_{1,4}&a_{1,3}\\
a_{2,2}&a_{2,1}&a_{2,4}&a_{2,3}\\
a_{3,2}&a_{3,1}&a_{3,4}&a_{3,3}
\end{array}\right]\text{, }A^{\dag}=\left[\begin{array}{ccc}
a_{1,2}&a_{2,2}&a_{3,2}\\
a_{1,1}&a_{2,1}&a_{3,1}\\
\hline
a_{1,4}&a_{2,4}&a_{3,4}\\
a_{1,3}&a_{2,3}&a_{3,3}
\end{array}\right].
\end{equation*}
\end{example}

In \cite{MarceloStandard}, Muniz proved that the generator matrix of an $[n\times s,k]$ NRT metric code $\mathcal{C}\subseteq M_{n,s}(\mathbb{F}_{q})$ is equivalent to an NRT-metric code with a special generator matrix.

\begin{theorem}\label{marceloS}(\cite{MarceloStandard}, Theorem 3.5). Let $\mathcal{C}\subseteq M_{n,s}(\mathbb{F}_{q})$ be an $[n\times s,k]$ NRT-metric code. Then $\mathcal{C}$ is equivalent to a code that has a generator matrix 
\begin{equation}
G=[G_{1}\mid G_{2}\mid\cdots\mid G_{n}],
\end{equation}
where
\begin{itemize}
    \item[i)] The nonzero rows of $G_{1}$ are distinct canonical vectors, arranged in order of increasing NRT weight.
    \item[ii)] $G$ is in the \textit{block echelon form}, i.e., if the last $t$ rows of $G_i$ are zero, then the last $t$ rows of $G_{1},\ldots,G_{i-1}$ are also zero.
\end{itemize}
\end{theorem}

Niederreiter-Rosenbloom-Tsfasman Linear complementary dual codes were characterized in terms of generator matrix by Heqian, Guangkui, and Wei in \cite{H.Xu}.

\begin{lemma}\label{GGo}(\cite{H.Xu})
Let $\mathcal{C}$ be an $[n\times s,k,d_{N}]$ NRT metric code. Let $G\in M_{k,ns}(\mathbb{F}_{q})$ be a generator matrix of $\mathcal{C}$
write as
\begin{equation}
G=[G_{1}\mid G_{2}\mid\cdots\mid G_{n}],
\end{equation}
where $G_{i}\in M_{k,s}(\mathbb{F}_{q})$ for $i=1,2,\ldots,n$. Then $\mathcal{C}$ is an NRT-LCD if and only if $GG^{\dag}$ is nonsingular.
\end{lemma}

\begin{example}
Let $\mathcal{C}$ be the binary $[1\times 4,2]$ code generate by 
\begin{equation*}
G=\left[\begin{array}{cccc}
1&0&1&1\\
0&1&1&1
\end{array}\right].
\end{equation*}
If $\mathcal{C}$ is considered as a linear code in the Hamming metric space $\mathbb{F}^{4}_{2}$. It can be check that $GG^{T}$ is nonsingular and then $\mathcal{C}$ is an Euclidean LCD code. However, $GG^{\dag}=0$ is singular. By Lemma \ref{GGo}, $\mathcal{C}$ is not an NRT-LCD code.
\end{example}

\begin{example}
Let $\mathcal{C}$ be the binary $[1\times 4,2]$ code generate by 
\begin{equation*}
G=\left[\begin{array}{cccc}
1&0&0&0\\
0&0&1&1
\end{array}\right].
\end{equation*}
If $\mathcal{C}$ is considered as a linear code in the Hamming metric space $\mathbb{F}^{4}_{2}$. It can be check that $GG^{T}=\left[\begin{array}{cc}
1&0\\
0&0 
\end{array}\right]$ and $GG^{T}$ is singular. That is, $\mathcal{C}$ is not an Euclidean LCD code. However, $GG^{\circ}=\left[\begin{array}{cc}
0&1\\
1&0 
\end{array}\right]$ and $GG^{\circ}$ is nonsingular. By Lemma \ref{GGo}, $\mathcal{C}$ is an NRT-LCD code.
\end{example}
\begin{definition}
Let $n,s,k$ be positive intergers. We define
\begin{equation}
LCD[n\times s,k]:=\max\lbrace d_{N}\mid\text{ there exist a binary $[n\times s,k]$ NRT-LCD code $\mathcal{C}\subseteq M_{n,s}(\mathbb{F}_{2})$}\rbrace.
\end{equation}
\end{definition}
\begin{remark}
Note that if $s=1$, then $$LCD[n\times 1,k]:=\max\lbrace d_{H}\mid \text{there exist a binary $[n,k]$ LCD Hamming metric code }\mathcal{C}\subseteq\mathbb{F}^{n}_{2}\rbrace.$$
So, we will focus on cases where $s>1$.
\end{remark}

\section{Elementary bounds on binary NRT-LCD codes in $M_{1,s}(\mathbb{F}_{2})$}\label{1xs}

Ozen and Siap in \cite{Ozen}, defined the generator matrix in the standard form for $[1\times s,k,d_{N}]$ NRT-metric codes $\mathcal{C}\subseteq M_{1,s}(\mathbb{F}_{q})$ as the following.
\begin{definition}(\cite{Ozen})\label{standard form}
Let $\mathcal{C}\subseteq M_{1,s}(\mathbb{F}_{q})$ be an $[1\times s,k,d_{N}]$ NRT-metric code and $G$ the generator matrix of $\mathcal{C}$. Applying certain elementary row operations one can always transform $G$ into the following form:
\begin{equation}\label{standardForm}
G^{\prime}=\left[\begin{array}{ccccccccccccc}
g_{1,1}&\cdots&g_{1,d_{1}-1}&g_{1,d_{1}}&0&\cdots&0&0&0&\cdots&0&\cdots&0\\
g_{2,1}&\cdots&g_{2,d_{1}-1}&0&g_{2,d_{1}+1}&\cdots&g_{1,d_{2}-1}&g_{2,d_{2}}&0&\cdots&0&\cdots&0\\
\vdots&\ddots&\vdots&\vdots&\vdots&\ddots&\vdots&\vdots&\vdots&\vdots&\vdots&\ddots&\vdots\\
g_{k,1}&\cdots&g_{k,d_{1}-1}&0&g_{k,d_{1}+1}&\cdots&g_{k,d_{2}-1}&0&g_{k,d_{2}+1}&\cdots&g_{k,d_{k}}&\cdots&0
\end{array}\right],
\end{equation}
where $g_{i,d_{i}}=1$, $d_{j,d_{i}}=0$ for $i\neq j$ and $\lbrace d_{1},d_{2},\ldots,d_{k}\rbrace$ is the set of $k$ possible nonzero NRT-weights so that $d_{N}=d_{1}<d_{2}<\ldots<d_{k}$.

Any generator matrix of the linear code $\mathcal{C}$ is equivalent to $G^{\prime}$. This $G^{\prime}$ is called the generator matrix in \textit{standard form}. A linear code having $G^{\prime}$ as its generator matrix in standard form is said to be of type $(d_{1},d_{2},\ldots,d_{k})$.
\end{definition}

\begin{remark}
Throughout this paper, we assume that $\mathcal{C}\subseteq M_{1,s}(\mathbb{F}_{q})$ always contains a codeword with full NRT-weight $s$, where $s$ is the length of the code. That is, the type of $\mathcal{C}$ is $(d_{N},d_{2},\ldots,s)$.
\end{remark}
\begin{lemma}
For any $s,k$ greater than zero,
\begin{equation}
LCD[1\times(s+2),k]\geq LCD[1\times s,k]+1.
\end{equation}
\end{lemma}
\begin{proof}
Let $G$ be a generator matrix in the standard form of an $[1\times s,k,d_{N}]$ NRT-LCD code $\mathcal{C}$. Then $GG^{\circ}$ is nonsingular since $\mathcal{C}$ is NRT-LCD. Let $\tilde{G}=\left[\begin{array}{c|c|c}
0&G&e_{k}
\end{array}\right]$, where $e^{T}_{k}$ is the given by $e^{T}_{k}=\left[\begin{array}{cccccc}
0&0&\cdots&0&1
\end{array}\right]$. Note that, $\tilde{G}$ is a generator matrix in the standard form of an $[1\times (s+2),k,d_{N}+1]$ NRT-metric code $\tilde{\mathcal{C}}$. Moreover, $$flip(\tilde{G})=\left[\begin{array}{c|c|c}
e_{k}&flip(G)&0
\end{array}\right]\text{ and }\tilde{G}^{\circ}=\left[\begin{array}{c}
e^{T}_{k}\\
\hline
G^{\circ}\\
\hline
0
\end{array}\right].$$
So,
\begin{equation}
\tilde{G}\tilde{G}^{\circ}=\left[\begin{array}{c|c|c}
0&G&e_{k}
\end{array}\right]\left[\begin{array}{c}
e^{T}_{k}\\
\hline
G^{\circ}\\
\hline
0
\end{array}\right]=GG^{\circ}.
\end{equation}
Hence $\tilde{G}\tilde{G}^{\circ}$ is nonsingular and $\tilde{G}$ generates an $[1\times (s+2),k,d_{N}+1]$ NRT-LCD code.
\end{proof}

\begin{proposition} For $s$ odd, and dimension $k=1$ or $k=s-1$ we have:
\begin{itemize}
\item[i)] $LCD[1\times s,1]=s$;
\item[ii)] $LCD[1\times s,s-1]=1$.
\end{itemize}
\end{proposition}
\begin{proof}
Singleton bound (\ref{Singletonbound}) says that any $[1\times s,1]$ NRT metric code has minimum distance upper bounded by $s$. We will prove that there exist an $[1\times s,1]$ NRT-LCD code $\mathcal{C}_{1}$ such that $d_{N}(\mathcal{C}_{1})=s$. Let $\mathcal{C}_{1}$ be the binary $[1\times s,1,s]$ NRT metric code with generator matrix given by $$G_{1}=\left[\begin{array}{ccc}
1&\cdots&1
\end{array}\right].$$
Note that $G_{1}^{\circ}=\left[\begin{array}{c}
1\\
\vdots\\
1\end{array}\right]$. So, $G_{1}G_{1}^{\circ}=1$ since $s$ is odd. That is, $G_{1}G_{1}^{\circ}$ is nonsingular and
$\mathcal{C}$ is an $[1\times s,1,s]$ NRT-LCD code. Proving that $LCD[1\times s,1]=s$.

Again, Singleton bound (\ref{Singletonbound}) says that any $[1\times s,s-1]$ NRT metric code has minimum distance upper bounded by $2$. We will prove that there exists an $[1\times s,s-1]$ NRT-LCD code $\mathcal{C}_{2}$ such that $d_{N}(\mathcal{C}_{2})=2$. Let $\mathcal{C}_{2}$ be the binary $[1\times s,s-1]$ NRT metric code with generator matrix given by the following standard form generator matrix
\begin{equation*}
G_{2}=\left[\begin{array}{cccccc}
1&1&0&0&\ldots&0\\
1&0&1&0&\ldots&0\\
1&0&0&1&\ldots&0\\
\vdots&\vdots&\vdots&\vdots&\ddots&\vdots\\
1&0&0&0&\ldots&1\\
\end{array}\right].
\end{equation*}
Then,
\begin{equation*}
G^{\circ}_{2}=\left[\begin{array}{ccccc}
0&0&0&\ldots&1\\
\vdots&\vdots&\vdots&\ddots&\vdots\\
0&0&1&\ldots&0\\
0&1&0&\ldots&0\\
1&0&0&\ldots&0\\
1&1&1&\ldots&1\\
\end{array}\right].
\end{equation*}
We can check that $G_{2}G^{\circ}_{2}$ is nonsingular and $\mathcal{C}_{2}$ is a binary $[1\times s,s-1,2]$ NRT-LCD code.
\end{proof}

\begin{theorem}
If $s$ is even, then there is no binary $[1\times s,1]$ NRT-LCD code.
\end{theorem}
\begin{proof}
Let $\mathcal{C}$ be a binary $[1\times s,1]$ NRT metric code. So, any generator matrix $G$ of $\mathcal{C}$ has the following form
\begin{equation*}
G=\left[\begin{array}{ccccc}
x_1&x_2&\cdots&x_{s-1}&x_s
\end{array}\right].
\end{equation*}
This implies that, $$G^{\circ}=\left[\begin{array}{l}
x_{s}\\
x_{s-1}\\
\vdots\\
x_{2}\\
x_{1}
\end{array}\right].$$

\noindent So, 
\begin{eqnarray*}
GG^{\circ}
&=&2\sum_{i=1}^{s}x_{i}x_{s-i+1}\\
&=&0.
\end{eqnarray*}
Proving that there is no binary $[1\times s,1]$ NRT-LCD code.
\end{proof}

\begin{corollary}
If $s$ is even, then there is no binary $[1\times s,s-1]$ NRT-LCD code.
\end{corollary}

\begin{theorem}
There is a binary maximum distance separable $[1\times s,2]$ NRT-LCD code.
\end{theorem}
\begin{proof}
If $\mathcal{C}$ is a binary maximum distance separable code of dimension two, then $d_{N}=s-1$. So, the generator matrix of $\mathcal{C}$ in standard form is
\begin{equation*}
G=\left[\begin{array}{ccccccc}
a_{1,1}&a_{1,2}&a_{1,3}&\cdots&a_{1,s-2}&1&0\\
a_{2,1}&a_{2,2}&a_{2,3}&\cdots&a_{2,s-2}&0&1
\end{array}\right],
\end{equation*}
where $a_{1,j},a_{2,j}\in\mathbb{F}_{2}$ for $1\leq j\leq s-2$. Then,
\begin{equation*}
G^{\circ}=\left[\begin{array}{cc}
0&1\\
1&0\\
a_{1,s-2}&a_{2,s-2}\\
\vdots&\vdots\\
a_{1,3}&a_{2,3}\\
a_{1,2}&a_{2,2}\\
a_{1,1}&a_{2,1}
\end{array}\right].
\end{equation*}
So, $GG^{\circ}=\left[\begin{array}{cc}
A_{1,1}&A_{1,2}\\
A_{2,1}&A_{2,2} 
\end{array}\right]$, where
\begin{itemize}
    \item[i)] If $s$ is even: $A_{1,1}=A_{2,2}=0$, and
    \begin{eqnarray*}
      A_{1,2}&=&a_{1,1}+a_{1,3}a_{2,s-2}+\ldots+a_{1,s-2}a_{2,3}+a_{2,2}\\
      A_{2,1}&=&a_{2,2}+a_{2,3}a_{1,s-2}+\ldots+a_{2,s-2}a_{1,3}+a_{1,1}.
    \end{eqnarray*}
Choosing $a_{1,1}=1$, and $a_{i,j}=0$ for $(i,j)\neq (1,1)$, we have that
\begin{equation*}
G=\left[\begin{array}{ccccccc}
1&0&0&\cdots&0&1&0\\
0&0&0&\cdots&0&0&1
\end{array}\right]
\end{equation*}
is the generator matrix of a binary maximum distance separable $[1\times s,2]$ NRT-LCD code, since $GG^{\circ}\left[\begin{array}{cc}
0&1\\
1&0 
\end{array}\right]$.

\item[ii)] If $n$ is odd: $A_{1,1}=a_{1,\frac{s+1}{2}}$, $A_{2,2}=a_{2,\frac{s+1}{2}}$, and
\begin{eqnarray*}
      A_{1,2}&=&a_{1,1}+a_{1,3}a_{2,s-2}+\ldots+a_{1,\frac{s+1}{2}}a_{2,\frac{s+1}{2}}+\ldots+a_{1,s-2}a_{2,3}+a_{2,2}\\
      A_{2,1}&=&a_{2,2}+a_{2,3}a_{1,s-2}+\ldots+a_{1,\frac{s+1}{2}}a_{2,\frac{s+1}{2}}+\ldots+a_{2,s-2}a_{1,3}+a_{2,2}.
\end{eqnarray*}
Choosing $a_{1,1}=a_{1,\frac{s+1}{2}}=a_{2,\frac{s+1}{2}}=1$ and $a_{1,j}=a_{2,j}=0$ for all $j\neq\frac{s+1}{2}$, we have that
\begin{equation*}
G=\left[\begin{array}{ccccccccccc}
1&0&0&\cdots&0&1&0&\cdots&0&1&0\\
0&0&0&\cdots&0&1&0&\cdots&0&0&1
\end{array}\right]
\end{equation*}
is the generator matrix of a binary maximum distance separable $[1\times s,2]$ NRT-LCD code, since $GG^{\circ}\left[\begin{array}{cc}
1&0\\
0&1 
\end{array}\right]$.
\end{itemize}
\end{proof}
\begin{corollary}
$LCD[1\times s,2]=s-1$ and $LCD[1\times s,s-2]=3$.
\end{corollary}
\begin{theorem}
There is no binary maximum distance separable $[1\times s,3]$ NRT-LCD code for even $s$.
\end{theorem}
\begin{proof}
If $\mathcal{C}$ is a binary maximum distance separable code of dimension 3, then $d_{N}=s-2$. So, the generator matrix of $\mathcal{C}$ in standard form is
\begin{equation*}
G=\left[\begin{array}{ccccccccc}
a_{1,1}&a_{1,2}&a_{1,3}&a_{1,4}&\cdots&a_{1,s-3}&1&0&0\\
a_{2,1}&a_{2,2}&a_{2,3}&a_{2,4}&\cdots&a_{2,s-3}&0&1&0\\
a_{3,1}&a_{3,2}&a_{3,3}&a_{3,4}&\cdots&a_{3,s-3}&0&0&1
\end{array}\right].
\end{equation*}
And,
\begin{equation*}
G^{\circ}=\left[\begin{array}{ccc}
0&0&1\\
0&1&0\\
1&0&0\\
a_{1,s-3}&a_{2,s-3}&a_{3,s-3}\\
a_{1,s-4}&a_{2,s-4}&a_{3,s-4}\\
\vdots&\vdots&\vdots\\
a_{1,4}&a_{2,4}&a_{3,4}\\
a_{1,3}&a_{2,3}&a_{3,3}\\
a_{1,2}&a_{2,2}&a_{3,2}\\
a_{1,1}&a_{2,1}&a_{3,1}
\end{array}\right].
\end{equation*}
So, $GG^{\circ}=\left[\begin{array}{ccc}
A_{1,1}&A_{1,2}&A_{1,3}\\
A_{2,1}&A_{2,2}&A_{2,3}\\
A_{3,1}&A_{3,2}&A_{3,3}
\end{array}\right]$, where $A_{1,1}=A_{2,2}=A_{3,3}=0$, and
\begin{eqnarray*}
 A_{1,2}=A_{2,1}&=&a_{1,2}+a_{2,3}+\sum_{j=4}^{s-3}a_{1,j}a_{2,s-j+1}\\
 A_{1,3}=A_{3,1}&=&a_{1,1}+a_{3,3}+\sum_{j=4}^{s-3}a_{1,j}a_{3,s-j+1}\\
 A_{2,3}=A_{3,2}&=&a_{2,1}+a_{3,2}+\sum_{j=4}^{s-3}a_{2,j}a_{3,s-j+1}.
\end{eqnarray*}
Finally, we can check that $\det(GG^{\circ})=0$. That is, $\mathcal{C}$ is not an $[1\times s,3]$ NRT-LCD code.
\end{proof}

\begin{theorem}
There is a binary maximum distance separable $[1\times s,3]$ NRT-LCD code for odd $s$.
\end{theorem}

\begin{proof}
Let $G$ be the generator matrix of $\mathcal{C}$ in the standard form. That is,

\begin{equation*}
G=\left[\begin{array}{ccccccccc}
a_{1,1}&a_{1,2}&a_{1,3}&a_{1,4}&\cdots&a_{1,s-3}&1&0&0\\
a_{2,1}&a_{2,2}&a_{2,3}&a_{2,4}&\cdots&a_{2,s-3}&0&1&0\\
a_{3,1}&a_{3,2}&a_{3,3}&a_{3,4}&\cdots&a_{3,s-3}&0&0&1
\end{array}\right].
\end{equation*}
Then,
\begin{equation*}
G^{\circ}=\left[\begin{array}{ccc}
0&0&1\\
0&1&0\\
1&0&0\\
a_{1,s-3}&a_{2,s-3}&a_{3,s-3}\\
a_{1,s-4}&a_{2,s-4}&a_{3,s-4}\\
\vdots&\vdots&\vdots\\
a_{1,4}&a_{2,4}&a_{3,4}\\
a_{1,3}&a_{2,3}&a_{3,3}\\
a_{1,2}&a_{2,2}&a_{3,2}\\
a_{1,1}&a_{2,1}&a_{3,1}
\end{array}\right].
\end{equation*}
So, $GG^{\circ}=\left[\begin{array}{ccc}
A_{1,1}&A_{1,2}&A_{1,3}\\
A_{2,1}&A_{2,2}&A_{2,3}\\
A_{3,1}&A_{3,2}&A_{3,3}
\end{array}\right]$, where $A_{i,i}=a_{i,\frac{s+1}{2}}$ for $i=1,2,3$, and
\begin{eqnarray*}
 A_{1,2}=A_{2,1}&=&a_{1,2}+a_{2,3}+\sum_{j=4}^{s-3}a_{1,j}a_{2,s-j+1}\\
 A_{1,3}=A_{3,1}&=&a_{1,1}+a_{3,3}+\sum_{j=4}^{s-3}a_{1,j}a_{3,s-j+1}\\
 A_{2,3}=A_{3,2}&=&a_{2,1}+a_{3,2}+\sum_{j=4}^{s-3}a_{2,j}a_{3,s-j+1}.
\end{eqnarray*}
We can check that $\det(GG^{\circ})=a_{1,\frac{s+1}{2}}a_{2,\frac{s+1}{2}}a_{3,\frac{s+1}{2}}+a_{1,\frac{s+1}{2}}A_{2,3}+a_{2,\frac{s+1}{2}}A_{1,3}+a_{3,\frac{s+1}{2}}A_{1,2}$. So, if we choose, for example, $a_{2,1}=a_{1,\frac{s+1}{2}}=1$, and $a_{i,j}=0$ for all $(i,j)\notin\lbrace (2,1)(1,\frac{s+1}{2})\rbrace$ then

\begin{equation*}
G=\left[\begin{array}{ccccccccccccc}
0&0&0&0&\cdots&0&1&0&\cdots&0&1&0&0\\
1&0&0&0&\cdots&0&0&0&\cdots&0&0&1&0\\
0&0&0&0&\cdots&0&0&0&\cdots&0&0&0&1
\end{array}\right]
\end{equation*}
is such that $\det(GG^{\circ})=1$. That is, $G$ is the generator matrix of an $[1\times s,3]$ NRT-LCD code.
\end{proof}
\begin{corollary}
If $s$ is odd, then $LCD[1\times s,3]=s-2$ and $LCD[1\times s,s-3]=4$.
\end{corollary}
In general, we have the following result for $k$ even.
\begin{theorem}\label{MDSexistenceforkeven}
There is a binary maximum distance separable $[1\times s,k]$ NRT-LCD code for $k$ is even. In other words, if $k$ is even then $LCD[1\times s,k]=s-k+1$.
\end{theorem}
\begin{proof}
Let $\mathcal{C}$ be the binary $[1\times s,k,s-k+1]$ NRT metric code given by the following generator matrix in the standard form
\begin{equation}
G=\left[\begin{array}{c}
e_{1}+e_{s-k+1}\\
e_{2}+e_{s-k+2}\\
\vdots\\
e_{\frac{k}{2}}+e_{s-\frac{k}{2}}\\
e_{s-\frac{k}{2}+1}\\
e_{s-\frac{k}{2}+2}\\
\vdots\\
e_{s-1}\\
e_{s}
\end{array}\right],
\end{equation}
where $e_{i}$ is the $i$-th canonical vector of $\mathbb{F}_{2}^{s}$ for $1\leq i\leq s$. So,
\begin{equation}
flip(G)=\left[\begin{array}{c}
e_{s}+e_{k}\\
e_{s-1}+e_{k-1}\\
\vdots\\
e_{s-\frac{k}{2}+1}+e_{\frac{k}{2}-1}\\
e_{\frac{k}{2}}\\
e_{\frac{k}{2}-1}\\
\vdots\\
e_{2}\\
e_{1}
\end{array}\right],
\end{equation}
and
\begin{equation}
G^{\circ}=\left[\begin{array}{c|c|c|c|c|c|c|c|c}
e^{T}_{s}+e^{T}_{k}&e^{T}_{s-1}+e^{T}_{k-1}&\cdots&e^{T}_{s-\frac{k}{2}+1}+e^{T}_{\frac{k}{2}-1}&e^{T}_{\frac{k}{2}}&e_{\frac{k}{2}-1}&\cdots&e^{T}_{2}&e^{T}_{1}
\end{array}\right].
\end{equation}
Finally, 
\begin{equation}
GG^{\circ}=\left[\begin{array}{ccccc}
0&0&\cdots&0&1\\
0&0&\cdots&1&0\\
\vdots&\vdots&\vdots&\vdots&\vdots\\
0&1&\cdots&0&0\\
1&0&\cdots&0&0\\
\end{array}\right].
\end{equation}
$GG^{\circ}$ is clearly nonsingular. Proving that $\mathcal{C}$ is a binary maximum distance separable $[1\times s,k]$ NRT-LCD code.
\end{proof}
The case where $s$ and $k$ are odd is described by the following theorem.
\begin{theorem}\label{MDSexistenceforskodd}
There is a binary maximum distance separable $[1\times s,k]$ NRT-LCD code for $s$ and $k$ odd. In other words, if $s$ and $k$ are odd then $LCD[1\times s,k]=s-k+1$.
\end{theorem}
\begin{proof}
Let $\mathcal{C}$ be the binary $[1\times s,k,s-k+1]$ NRT metric code given by the following generator matrix in the standard form
\begin{equation}
G=\left[\begin{array}{c}
e_{\frac{s+1}{2}}+e_{s-k+1}\\
e_{1}+e_{s-k+2}\\
e_{2}+e_{s-k+3}\\
\vdots\\
e_{\frac{k+1}{2}-1}+e_{s-k+\frac{k+1}{2}}\\
e_{s-k+\frac{k+1}{2}}\\
e_{s-k+\frac{k+1}{2}+1}\\
\vdots\\
e_{s-1}\\
e_{s}
\end{array}\right],
\end{equation}
where $e_{i}$ is the $i$-th canonical vector of $\mathbb{F}_{2}^{s}$ for $1\leq i\leq s$. So,
\begin{equation}
flip(G)=\left[\begin{array}{c}
e_{\frac{s+1}{2}}+e_{k}\\
e_{s}+e_{k-1}\\
e_{s-1}+e_{k-2}\\
\vdots\\
e_{s-\frac{k-3}{2}}+e_{\frac{k+1}{2}}\\
e_{\frac{k+1}{2}-1}\\
e_{\frac{k+1}{2}-2}\\
\vdots\\
e_{2}\\
e_{1}
\end{array}\right],
\end{equation}
and
\begin{equation}
G^{\circ}=\left[\begin{array}{c|c|c|c|c|c|c|c|c|c}
e^{T}_{\frac{s+1}{2}}+e^{T}_{k}&e^{T}_{s}+e^{T}_{k-1}&e^{T}_{s-1}+e^{T}_{k-2}&\cdots&e^{T}_{s-\frac{k-3}{2}}+e^{T}_{\frac{k+1}{2}}&e^{T}_{\frac{k+1}{2}-1}&e_{\frac{k+1}{2}-2}&\cdots&e^{T}_{2}&e^{T}_{1}
\end{array}\right].
\end{equation}
Finally, 
\begin{equation}
GG^{\circ}=\left[\begin{array}{ccccc}
1&0&\cdots&0&0\\
0&0&\cdots&0&1\\
0&0&\cdots&1&0\\
\vdots&\vdots&\vdots&\vdots&\vdots\\
0&1&\cdots&0&0\\
\end{array}\right].
\end{equation}
$GG^{\circ}$ is clearly nonsingular. Proving that $\mathcal{C}$ is a binary maximum distance separable $[1\times s,k]$ NRT-LCD code.
\end{proof}

\section{Elementary bounds on binary NRT-LCD codes $M_{n,2}(\mathbb{F}_{2})$}\label{nx2}
In this section, we prove a condition on the non-existence of binary maximum distance separable $[n\times 2,k]$ NRT-LCD codes. For this purpose, we will consider the Plotkin bound \cite{NRT} for NRT metric codes.
\begin{theorem}(\cite{NRT}\label{Plotkin}, Plotkin bound for NRT metric codes). Let $\mathcal{C}$ be an $[n\times s,k,d_{N}]$ NRT metric code such that $d_{N}>ns\delta_{crit}$, where $\delta_{crit}=1-\frac{1}{sq^{s}}\frac{q^{s}-1}{q-1}$. Then,
\begin{equation}\label{Plotkin}
    q^{k}\leq\frac{d_{N}}{d_{N}-ns\delta_{crit}}.
\end{equation}
\end{theorem}

\begin{lemma}
$LCD[n\times 2,k]\leq\left(\frac{5}{2^{k}-1}\right)n<2n-k+1$ for $n>\left(\frac{2^{k}-1}{2^{k+1}-7}\right)(k-1)$.
\end{lemma}
\begin{proof}
Let's first prove that $LCD[n\times 2,k]\leq\left(\frac{5}{2^{k}-1}\right)n$. By the Plotkin bound (\ref{Plotkin}), any binary $[n\times 2,k,d_{N}]$ NRT metric code with $d_{N}>2n\delta_{crit}=\frac{5}{4}n$ satisfies
\begin{equation}
2^{k}\leq\frac{d_{N}}{d_{N}-\frac{5n}{4}}.
\end{equation}
Solving for $d_{N}$, we have $d_{N}\leq\left(\frac{5}{2^{k}-1}\right)n$. Therefore any $[n\times 2, k, d_{N}]$ NRT-LCD code must satisfy this inequality. proving that $LCD[n\times 2,k]\leq\left(\frac{5}{2^{k}-1}\right)n$. Now, note that
\begin{equation}
2n-k+1-\left(\frac{5}{2^{k}-1}\right)n>0\Longleftrightarrow\left(\frac{2^{k+1}-7}{2^{k}-1}\right)>k-1\Longleftrightarrow n>\left(\frac{2^{k}-1}{2^{k+1}-7}\right)(k-1).
\end{equation}
Proving that $\left(\frac{5}{2^{k}-1}\right)n<2n-k+1$ if $n>\left(\frac{2^{k}-1}{2^{k+1}-7}\right)(k-1)$.
\end{proof}

\begin{corollary}
Let $n,k$ be positive integers such that $n>\left(\frac{2^{k}-1}{2^{k+1}-7}\right)(k-1)$. Then, there is no binary maximum distance separable $[n\times 2, k]$ NRT-LCD code.
\end{corollary}

\begin{corollary}
Let $n,k$ be positive integers such that $n>\left(\frac{2^{k}-1}{2^{k+1}-7}\right)(k-1)$. Then, $$LCD[n\times 2, k]<2n-k+1.$$
\end{corollary}

\begin{theorem}
We have the following:
\begin{itemize}
    \item[i)] $LCD[2\times 2,2]=2$;
    \item[ii)] $LCD[3\times 2,2]=5$;
\end{itemize}
\end{theorem}
\begin{proof}
i) Let $\mathcal{C}$ be the binary $[2\times 2, 2]$ NRT metric code given by the following generator matrix
\begin{equation*}
    G=\left[\begin{array}{cc|cc}
    1&0&1&1\\
    0&1&1&1
    \end{array}\right].
\end{equation*}
It is easy to check that $GG^{\dag}$ is non-singular and $d_{N}(\mathcal{C})=2$. 

Let's prove that there is no binary $[2\times 2,2,3]$ NRT-LCD code. Note that by Theorem \ref{marceloS}, if $\mathcal{C}$ is a binary $[2\times 2, 2]$ NRT metric code, then $\mathcal{C}$ is equivalent to a code with generator matrix given by
$G=[G_{1}\mid G_{2}]$, where the nonzero rows of $G_1$ are distinct canonical vectors arranged in order of increasing NRT weight, and $G$ is in the block echelon form. Moreover, the last row of $G_{2}$ is non-zero, since $\dim(\mathcal{C})=2$.

If the last row of $G_{1}$ is zero then $\mathcal{C}$ has a codeword with a zero row. That is, in this case $d_{N}(\mathcal{C})\leq 2$. 

So, the only case that we have to analyze is the case where
\begin{equation}\label{Gen22}
    G=\left[\begin{array}{cc|cc}
    1&0&a&b\\
    0&1&c&d
    \end{array}\right],
\end{equation}
where $(c,d)\neq (0,0)$. In this case,
\begin{equation*}
    GG^{\dag}=\left[\begin{array}{cc}
    0&1+ad+bc\\
    1+ad+bc&0
    \end{array}\right]
\end{equation*}
and the generator matrices $G$ of the form (\ref{Gen22}) such that $GG^{\dag}$ is non-singular are
\begin{equation*}
    \left[\begin{array}{cc|cc}
    1&0&1&0\\
    0&1&1&0
    \end{array}\right],\left[\begin{array}{cc|cc}
    1&0&0&1\\
    0&1&0&1
    \end{array}\right],\left[\begin{array}{cc|cc}
    1&0&0&0\\
    0&1&1&1
    \end{array}\right],\left[\begin{array}{cc|cc}
    1&0&0&0\\
    0&1&0&1
    \end{array}\right].
\end{equation*}
The one can check that those matrices generates $[2\times 2,2]$ NRT-LCD codes with $d_{N}(\mathcal{C})\leq 2$.

\noindent ii) Let $\mathcal{C}$ be the binary $[2\times 2,2]$ NRT-metric code with generator matrix given by

\begin{equation*}
   G=\left[\begin{array}{cc|cc|cc}
    1&0&1&1&1&1\\
    0&1&0&1&1&0
    \end{array}\right].
\end{equation*}
It is easy to check that $GG^{\dag}=\left[\begin{array}{cc}
    0&1\\
    1&0
\end{array}\right]$. That is, $\mathcal{C}$ is an binary $[3\times 2,2]$ NRT-LCD code. Moreover,
\begin{equation*}
\mathcal{C}=\left\lbrace\left[\begin{array}{cc}
    0&0\\
    0&0\\
    0&0
\end{array}\right],\left[\begin{array}{cc}
    1&0\\
    1&1\\
    1&1
\end{array}\right],\left[\begin{array}{cc}
    0&1\\
    0&1\\
    1&0
\end{array}\right],\left[\begin{array}{cc}
    1&1\\
    1&0\\
    0&1
\end{array}\right]\right\rbrace
\end{equation*}
and $d_{N}(\mathcal{C})=5$.
\end{proof}

\begin{corollary}\label{ndif3}
There is a binary maximum distance separable $[n\times 2,2]$ NRT-LCD code if and only if $n=3$. 
\end{corollary}

\begin{theorem}
Let $n\neq 3$ be a positive integer. Then
$$n\leq LCD[n\times 2,2]\leq 2n-1.$$
\end{theorem}
\begin{proof}
$LCD[n\times 2,2]\leq 2n-1$ by Corollary \ref{ndif3}. Let's check that $LCD[n\times 2,2]\geq n$. Note that, if $n$ is odd, then the binary $[n\times 2, 2]$ NRT-LCD code
\begin{equation*}
\mathcal{C}=\left\lbrace 0,\left[\begin{array}{c|c}
\sum_{i=1}^{n}e^{T}_{i}&0^{T}
\end{array}\right],\left[\begin{array}{c|c}
0^{T}&\sum_{i=1}^{n}e^{T}_{i}
\end{array}\right],\left[\begin{array}{c|c}
\sum_{i=1}^{n}e^{T}_{i}&\sum_{i=1}^{n}e^{T}_{i}
\end{array}\right]\right\rbrace.
\end{equation*}
is such that $d_{N}(\mathcal{C})=n$.

Now, if $n$ is even, then the binary $[n\times 2,2]$ NRT-LCD code
\begin{equation*}
\mathcal{C}=\left\lbrace 0,\left[\begin{array}{c|c}
\sum_{i=1}^{n}e^{T}_{i}&0^{T}
\end{array}\right],\left[\begin{array}{c|c}
0^{T}&\sum_{i=1}^{n-1}e^{T}_{i}
\end{array}\right],\left[\begin{array}{c|c}
\sum_{i=1}^{n}e^{T}_{i}&\sum_{i=1}^{n-1}e^{T}_{i}
\end{array}\right]\right\rbrace.
\end{equation*}
is such that $d_{N}(\mathcal{C})=n$.
\end{proof}

\subsection{A linear programming bound for binary NRT-LCD codes in $M_{n,2}(\mathbb{F}_{2})$}
In this section, we present a linear programming bound for binary NRT-LCD codes. We will make use of generalized MacWilliams identity for NRT metric codes. The MacWilliams identity was proven with different approaches in \cite{BargAndPurka, H.Trinker,Siap,DS}. We will present this identity for the case where $s=2$ to suit to the context of the present paper.

\begin{definition}
Let $v\in M_{n,2}(\mathbb{F}_{2})$ written as $[v_{1};v_{2};\ldots,v_{n}]$, where $v_{i}=[v_{i,1},v_{i,2}]$ for $i=1,\ldots, n$. The shape of $v$ with respect to the NRT weight is an $2$-vector $shape(v)=(e_{1},e_{2})$, where
\begin{equation}
e_{j}=\mid\lbrace i: 1\leq i\leq n\text{ and }\rho(v_{i})=j\rbrace\mid.
\end{equation}
Also define $e_{0}=n-\mid shape(v)\mid$, where $\mid shape(v)\mid=e_{1}+e_{2}$.
\end{definition}
A combinatorial calculation shows that the number og $v\in M_{n,2}(\mathbb{F}_{2})$ such that $shape(e_1,e_{2})$ is given by
\begin{equation}
\binom{n}{e_{0},e_{1},e_{2}}2^{e_{2}}.
\end{equation}

Note that, if for $v\in M_{n,2}(\mathbb{F}_{2})$, we define $\mid shape(e)\mid^{\prime}=e_{1}+2e_{2}$, then $\rho(v)=\mid shape(e)\mid^{\prime}$.

The shape enumerator of $\mathcal{C}\subseteq M_{n,2}(\mathbb{F}_{2})$ is the polynomial 
\begin{eqnarray*}
 H_{\mathcal{C}}(z_{0},z_{1},z_{2})=\sum_{e\in\Delta_{n,2}}\mathcal{A}_{e}z^{e_0}_{0}z^{e_1}_{1}z^{e_2}_{2},
\end{eqnarray*}
where $\mathcal{A}_{e}:=\mid\lbrace v\in\mathcal{C}:shape(v)=e\rbrace\mid$ and $\Delta_{n,2}=\lbrace e\in\mathbb{N}^{2}:e_{0}+e_{1}+e_{2}=n\rbrace$.

\begin{theorem}(\cite{H.Trinker})\label{MacWilliamsShape}
Let $\mathcal{C}\subseteq M_{n,2}(\mathbb{F}_{2})$ be an $[n\times 2,k]$ NRT metric code. Further on let $\left(\mathcal{A}_{e}\right)$ be the shape distribution of $\mathcal{C}$ and $\left(\mathcal{A}^{\prime}_{e}\right)$ be the shape distribution of $\mathcal{C}^{\perp_{N}}$. Then
\begin{equation*}
    \mathcal{A}^{\prime}_{e}=\frac{1}{2^{k}}\sum_{e^{\star}\in\Delta_{n,2}}B_{e}(e^{\star})\mathcal{A}_{e^{\star}},
\end{equation*}
where
\begin{equation*}
B_{e}(e^{\star})=2^{e_{2}}\prod^{2}_{t=1}\left(\sum_{j_{t}=0}^{e_t}(-1)^{j_t}\binom{e^{\star}_{3-t}}{j_{t}}\binom{\sum_{r=0}^{2-t}e^{\star}_{r}-\sum_{r=t+1}^{2}e^{\star}_{r}}{e_{t}-j_{t}}\right).
\end{equation*}
\end{theorem}

Note that if $\mathcal{C}$ is a binary $[n\times 2,k]$ NRT-LCD code and $\mathcal{A}_{e}$ its shape distribution. Let $\mathcal{A}_{e}^{\prime}$ be the shape distribution of its dual code $\mathcal{C}^{\perp_{N}}$. By the definition of NRT-LCD code, if $v\in\mathcal{C}$ then $v\notin\mathcal{C}^{\perp_{N}}$, so
\begin{equation}\label{A_e+A_eP}
\mathcal{A}_{e}+\mathcal{A}^{\prime}_{e}\leq\binom{n}{e_{0},e_{1},e_{2}}2^{e_{2}}.
\end{equation}
By the Theorem \ref{MacWilliamsShape}, we have that 
\begin{equation}\label{A_eP}
\mathcal{A}^{\prime}_{e}=\frac{1}{2^{k}}\sum_{e^{\star}\in\Delta_{n,2}}B_{e}(e^{\star})\mathcal{A}_{e^{\star}}.
\end{equation}
Moreover, since $\dim\mathcal{C}=k$, we have that 
\begin{equation}\label{2k}
2^{k}=\sum_{e^{\star}\in\Delta_{n,2}}\mathcal{A}_{e^{\star}}.
\end{equation}
Combining (\ref{A_e+A_eP}), (\ref{A_eP}) and (\ref{2k}). We have the following bound for binary NRT-LCD codes.

\begin{theorem}
If $\mathcal{C}$ is a binary $[n\times 2,k]$ NRT-LCD code. Then,
\begin{equation}
2^{k}\mathcal{A}_{e}\leq\sum_{e^{\star}\in\Delta_{n,2}}\mathcal{A}_{e
^{\star}}\left[\binom{n}{e_{0},e_{1},e_{2}}-B_{e}(e^{\star})\right].
\end{equation}
\end{theorem}

\section{Constructions}\label{constructions}
In this section, we present two new constructions of binary Niederreiter-Rosenbloom-Tsfasman linear complementary dual codes.

The first construction (Theorem \ref{construction1}) gives a binary NRT-LCD code in $M_{1,s+2}(\mathbb{F}_{2})$ from a binary NRT-LCD code in $M_{1,s+2}(\mathbb{F}_{2})$, where $s$ is odd. The second construction (Theorem \ref{construction2}) gives a binary NRT-LCD code in $M_{n,s+2}(\mathbb{F}_{2})$,  from a binary NRT-LCD code in $M_{1,s}(\mathbb{F}_{2})$, where $s$ and $n$ are odd.

\begin{theorem}\label{construction1}
Let $\mathcal{C}$ be a binary $[1\times s,k]$ NRT-LCD code with generator matrix
\begin{equation}
G:=\left[\begin{array}{c}
g_{1}\\
g_{2}\\
\vdots\\
g_{k}
\end{array}\right]=\left[\begin{array}{cccc}
g_{1,1}&g_{1,2}&\ldots&g_{1,s}\\
g_{2,1}&g_{2,2}&\ldots&g_{2,s}\\
\vdots&\vdots&\ddots&\vdots\\
g_{k,1}&g_{k,2}&\ldots&g_{k,s}
\end{array}\right].
\end{equation}
Let $x=(x_1,x_2,\ldots,x_s)\in\mathbb{F}^{s}_{2}$ and $\mathcal{C}(x)$ be the binary code with the following generator matrix
\begin{equation}
G(x)=\left[\begin{array}{l|lll|l}
1&x_{1}&\ldots&x_{s}&0\\
\hline
\langle x,g_1\rangle_{N}&g_{1,1}&\ldots&g_{1,s}&\langle x,g_1\rangle_{N}\\
\vdots&\vdots&\ddots&\vdots&\vdots\\
\langle x,g_k\rangle_{N}&g_{k,1}&\cdots&g_{k,s}&\langle x,g_k\rangle_{N}
\end{array}\right].
\end{equation}
If $s$ is odd and $x_{\frac{s+1}{2}}=1$ and $x_{j}=0$ for $ j>\frac{s+1}{2}$. Then, $\mathcal{C}(x)$ is a binary $[1\times (s+2),k+1]$ NRT-LCD code.
\end{theorem}
\begin{proof}
It is clear that $\mathcal{C}(x)$ is a binary $[1\times(s+2),k+1]$ NRT metric code.
Note that
\begin{equation}
G^{\circ}(x)=\left[\begin{array}{l|llll}
0&\langle x,g_1\rangle_{N}&\ldots&\langle x,g_{k-1}\rangle_{N}&\langle x,g_{k}\rangle_{N}\\
\hline
x_{s}&g_{2,s}&\ldots&g_{k-1,s}&g_{k,s}\\
x_{s-1}&g_{2,s-1}&\ldots&g_{k-1,s-1}&g_{k,s-1}\\
\vdots&\vdots&\ddots&\vdots&\vdots\\
x_{1}&g_{1,1}&\ldots&g_{k-1,1}&g_{k,1}\\
\hline
1&\langle x,g_{1}\rangle_{N}&\cdots&\langle x,g_{k}\rangle_{N}&\langle x,g_{k}\rangle_{N} 
\end{array}\right].
\end{equation}
Thus,
\begin{equation*}
G(x)G^{\circ}(x)=\left[\begin{array}{c|c}
\langle x,x\rangle_{N}&\textbf{0}\\
\hline
&\multirow{2}{2cm}{$GG^{\circ}$}\\
\textbf{0}^{T}& \\
\end{array}\right],\text{ where }\textbf{0}=(0,0,\ldots,0)\in\mathbb{F}^{k}_{2}.\\
\end{equation*}
Note also that $\langle x,x\rangle_{RT}=\left(x_{\frac{s+1}{2}}\right)^2=1$. That is, $G(x)G^{\circ}(x)$ is nonsigular. Proving that $\mathcal{C}(x)$ is a binary $[1\times(s+2),k+1]$ NRT-LCD code.
\end{proof}

\begin{theorem}\label{construction2}
Let $\mathcal{C}$ be a binary $[1\times s,k]$ NRT-LCD code with generator matrix
\begin{equation}
G:=\left[\begin{array}{c}
g_{1}\\
g_{2}\\
\vdots\\
g_{k}
\end{array}\right]=\left[\begin{array}{cccc}
g_{1,1}&g_{1,2}&\ldots&g_{1,s}\\
g_{2,1}&g_{2,2}&\ldots&g_{2,s}\\
\vdots&\vdots&\ddots&\vdots\\
g_{k,1}&g_{k,2}&\ldots&g_{k,s}
\end{array}\right].
\end{equation}
Let $x=(x_1,x_2,\ldots,x_s)\in\mathbb{F}^{s}_{2}$, and $\mathcal{C}(x)$ be the binary $[1\times(s+2),k+1]$ NRT code with the following generator matrix
\begin{equation}
G(x)=\left[\begin{array}{l|lll|l}
1&x_{1}&\ldots&x_{s}&0\\
\hline
\langle x,g_1\rangle_{N}&g_{1,1}&\ldots&g_{1,s}&\langle x,g_1\rangle_{N}\\
\vdots&\vdots&\ddots&\vdots&\vdots\\
\langle x,g_k\rangle_{N}&g_{k,1}&\cdots&g_{k,s}&\langle x,g_k\rangle_{N}
\end{array}\right].
\end{equation}
If $s$ and $n$ are odd, $x_{\frac{s+1}{2}}=1$ and $x_{j}=0$ for $ j>\frac{s+1}{2}$. Then, the code $\mathcal{C}_{n}(x)$ generated by
\begin{equation}
G_{(n)}:=G(x,x,\ldots,x)=[G(x)\mid G(x)\mid\cdots\mid G(x)]
\end{equation}
is a binary $[n\times(s+2),k+1]$ NRT-LCD code.
\end{theorem}

\begin{proof}
It is clear that the length of $\mathcal{C}_{n}(x)$ is a binary $[n\times (s+2),k+1]$ NRT metric code.
Note that
\begin{equation}
G_{(n)}^{\dag}=\left[\begin{array}{c}
 G^{\circ}(x)\\
 G^{\circ}(x)\\
 \vdots\\
 G^{\circ}(x)
\end{array}\right].
\end{equation}
Thus,
\begin{eqnarray*}
G_{(n)}G^{\dag}_{(n)}&=&G(x)G^{\circ}(x)+G(x)G^{\circ}(x)+\cdots+G(x)G^{\circ}(x)\\
&=&G(x)G^{\circ}(x)
\end{eqnarray*}
since $n$ is odd. So, $G_{(n)}G^{\dag}_{(n)}$ is nonsingular. Proving that $\mathcal{C}_{n}(x)$ is a binary $[1\times(s+2),k+1]$ NRT-LCD code.
\end{proof}

\section{Conclusion}

In this paper, we studied the maximum minimum distance of a binary NRT-LCD code. In particular, we proved the existence of binary maximum distance separable $[1\times s, k]$ NRT-LCD codes for the cases where $k$ is even or $s$ and $k$ are odd. We derived conditions on $n$ for the existence of binary $[n\times 2,k]$ NRT-LCD codes. Finally, we present a linear programming bound for binary $[n\times 2,k]$ NRT-LCD codes and two new constructions. It would be interesting to find construction methods of maximum distance separable codes with small hulls for the parameters that we can not construct maximum distance separable NRT-LCD code.

\end{document}